\documentclass{article}
\usepackage{macros}
\usepackage[margin=1.25in]{geometry}
\usepackage{physics}
\usepackage{quantikz}
\usepackage{authblk}

\usepackage{amsthm,amsmath,amsfonts,amssymb}

\newcommand{\CNOT}{\mathrm{CNOT}}
\newcommand{\SWAP}{\mathrm{SWAP}}

\usepackage[backend=biber,
            style=trad-alpha,
            sorting=nty,
            doi=true]{biblatex}
\usepackage{csquotes}
\bibliography{refs}

\title{Fast Cliffords When Your Quantum Memory Is Full}

\author[1]{Marten Folkertsma\thanks{\texttt{m.j.folkertsma@uva.nl}}}
\author[2]{Ian Mertz\thanks{\texttt{iwmertz@iuuk.mff.cuni.cz}}}
\author[3]{Sergii Strelchuk\thanks{\texttt{Sergii.Strelchuk@cs.ox.ac.uk}}}
\author[3]{Sathyawageeswar Subramanian\thanks{\texttt{Sathya.Subramanian@cs.ox.ac.uk}}}
\affil[1]{\small \textit{University of Amsterdam, QuSoft}}
\affil[2]{\small \textit{Charles University, Prague, Czech Republic}}
\affil[3]{\small \textit{Department of Computer Science, University of Oxford, Parks Rd, Oxford OX1 3QG, United Kingdom}}

\date{}

\begin{document}

\maketitle

\begin{abstract}
Additional qubits can reduce the depth of a quantum circuit by providing workspace for parallel computation, but standard constructions assume that this workspace is initialized in a known state. 
In this work we study catalytic implementations, i.e. asking whether dirty qubits can instead be used provided that their joint state including any entanglement with other registers is restored exactly at the end of the computation.

We show that every $n$-qubit Clifford circuit has a catalytic implementation
of depth $O(\log n)$ using $O(n^2/\log^2 n)$ catalytic qubits and no clean
qubits, matching the asymptotic depth achievable when clean workspace is available.
We extend this
approach to diagonal elements of any fixed level $C_k$ of the Clifford
hierarchy, which admit catalytic implementations of depth $O(\log(n+1))$
with $O(n^k/\log(n))$ gates and $O(n^k/\log^2(n))$ catalytic qubits, as well as to semi-Clifford Gates.

\end{abstract}

\setcounter{tocdepth}{2}
\tableofcontents

\section{Introduction}

Auxiliary qubits are a basic resource in quantum computing. By storing intermediate information, they allow different parts of a computation to proceed in parallel, and many circuit-synthesis results achieve depth reduction by trading in additional workspace for fewer sequential layers. This workspace, however, is typically assumed to be initialized in a known state. This assumption introduces an important distinction between available qubits from available temporary memory. In particular, a device may contain many idle qubits that cannot be treated as \emph{clean} ancillas for use as workspace because they already store quantum information that must be preserved.

This motivates a different notion of workspace, namely qubits that may already contain arbitrary quantum information, but can be used during the computation provided that their state is restored exactly at the end. We call such workspace catalytic. 

Computing with non-initialized memory has been investigated in several forms. 
In space-bounded computation, catalytic memory has developed into a broader model in which a computation may use a full auxiliary memory provided that it is restored at the end \cite{BuhrmanCleveKouckyLoffSpeelman14}, with subsequent work showing that such memory can substitute for or augment other computational resources and can yield nontrivial space and time savings \cite{AgarwalaMertz25,CookLiMertzPyne25,KouckyMertzPyneSami25, CookPyne26,ChakrabortyDattaKusreMukhopadhyaySinhababu26,ChmelDudejaKouckyMertzRajgopal26}. A quantum analogue of catalytic space was introduced in \cite{BuhrmanFolkertsmaMertzSpeelmanStrelchukSubramanianTupker25}, and related in-place techniques have also been used in quantum circuit constructions~\cite{Remaud}. 

In a related line of work, quantum algorithms that do not require initialization of their auxiliary
registers have been studied for problems including Deutsch–Jozsa, Simon’s problem, and period finding~\cite{Chi01,Chi05}, and Takahashi and Tani have studied the power of uninitialized qubits in shallow quantum circuits~\cite{TakahashiTani2021}. These models allow auxiliary qubits to start in unknown states, but do not require them to be restored for every initial state, in contrast to the quantum catalytic space model of \cite{BuhrmanFolkertsmaMertzSpeelmanStrelchukSubramanianTupker25}. Thus it remains less understood whether such information-bearing auxiliary memory can replace the clean ancillas used to reduce the depth of a prescribed quantum circuit, while preserving every
state of the auxiliary register. 

For $\CNOT$ and Clifford circuits, Jiang et
al.~\cite{doi:10.1137/1.9781611975994.13} established an optimal trade-off
between depth and the number of initialized auxiliary qubits, showing in particular that sufficiently many clean ancillas can be used to reduce arbitrary $n$-qubit Clifford circuits to logarithmic depth. Their construction therefore gives a precise benchmark for the power of clean workspace. We ask whether the same depth reduction remains possible when the initialized ancillas are replaced entirely by catalytic qubits whose initial state is unknown, arbitrary, and must be preserved.

Recent work already gives partial affirmative answers to this question. Li, Tian, He, and Sun show that any $n$-qubit CNOT circuit can be implemented in depth $O(\log n)$ using $O(n^2)$ dirty ancillas, via a decomposition into bipartite CNOT circuits~\cite{LiTianHeSun2025hammingweight}. Related recent work of Du, Cheng, and Ma develops a general catalytic compiler for bounded-fan-in, bounded-fan-out XOR computations, preserving their asymptotic depth while replacing clean intermediate storage by dirty workspace~\cite{DuChengMa2026lowdepthrandomunitaries}. Our results complement these recent developments by recovering the full depth--workspace tradeoff available with initialized ancillas: every $n$-qubit Clifford circuit admits a catalytic implementation of depth $O\left(\frac{n}{s\log n}\right)$ using $O(sn)$ catalytic qubits, for $1\leq s\leq n/\log^2 n$. In particular, achieving depth $O(\log n)$ requires only $O\left(n^2/\log^2 n\right)$ catalytic qubits and no clean qubits. 

Our approach is structurally different from the preceding constructions. Starting from the parallel CNOT compiler of Jiang et al, we cycle the catalytic register through a fixed family of linear transformations whose contributions cancel the dependence on its unknown initial state, and then use a factorization due to Urschel~\cite{urschel2023} along with the Clifford normal form of Aaronson
and Gottesman~\cite{AaronsonGottesman2004} to eliminate the clean output register altogether. This approach allows us to extend to diagonal gates in the Clifford hierarchy using
the classification of Cui, Gottesman, and
Krishna~\cite{CuiGottesmanKrishna2017} and exact phase synthesis based on
Barenco et al.~\cite{BarencoEtAl1995}.

These results show that clean workspace can be replaced by catalytic memory for broad classes of quantum circuits, while retaining the depth advantages provided by initialized workspace. They also suggest a broader question: how far can this replacement principle be pushed for general quantum computation?

\subsection{Our results}
Our main theorem shows that the depth--workspace tradeoff for Clifford circuits can be achieved using catalytic workspace in place of clean auxiliary qubits.

\begin{theorem}
\label{thm:catalytic-clifford}
Let $n\geq 2$ and let $s$ be an integer with
$1\leq s\leq n/\log^2 n$.
For every $n$-qubit Clifford circuit $\mathcal C$, there exists
a catalytic Clifford circuit $\mathcal C_c$ of depth
$O\!\left(\frac{n}{s\log n}\right)$
using $O(sn)$ catalytic qubits and no clean qubits.
\end{theorem}

As an application of these catalytic circuits, we show in
Section~\ref{sec:Toffoli_construction} that the Toffoli construction
of~\cite{GKZ25}, which uses exponentially fewer $\mathsf{T}$ gates at the
cost of being approximate, can be made optimal in depth in the presence of catalytic qubits.

\begin{theorem}[Informal; see Theorem~\ref{thm:catalytic-or}]
  An $\epsilon$-approximate $n$-qubit Toffoli gate can be implemented in
  depth $O(\log n)$ using $O(\log(1/\epsilon))$ $\mathsf{T}$ gates,
  $O(n\log(1/\epsilon))$ catalytic qubits and $O(\log(1/\epsilon))$ clean
  qubits.
\end{theorem}

The same catalytic depth reduction extends beyond Clifford circuits to diagonal gates in the Clifford hierarchy. For a fixed level $k$, the phase of such a gate is described by a polynomial of degree at most $k$, and we exploit this bounded structure to cancel the dependence on the initial catalyst state. This gives logarithmic depth with polynomial catalytic workspace.

\begin{theorem}[{See Theorem~\ref{thm:catalytic-diagonal}}]
\label{thminf:catalytic-diagonal}
Fix $k\geq 2$, and let $s\geq 1$ be an integer. For every $n$-qubit diagonal unitary $D\in\mathcal{C}_k$ in the $k$th level of the Clifford hierarchy, there is a catalytic Clifford circuit $\mathcal C_D$ over $G$ using $O(sn)$ catalytic qubits and no clean qubits of depth $O\left(\log(n+1)+\frac{n^{k-1}}{s\log(n+1)}\right)$. In particular, there is an implementation of depth $O(\log(n+1))$ using $O(n^k/\log^2(n+1))$ catalytic qubits, with size $ O\left(\frac{n^k}{\log(n+1)}\right)$.
\end{theorem}

We supplement this in Lemma~\ref{lem:diagonal-lower-bounds} with a lower bound showing that our construction achieves the optimal depth. The construction also yields the same asymptotic bounds for semi-Clifford unitaries given a Clifford--diagonal decomposition (see Corollary~\ref{cor:catalytic-semi-clifford}). 

When catalytic qubits are abundant, the remaining depth overhead of $O(\log(n))$
comes from implementing the fan-out gate. With native access to the fan-out gate,
the depth of these results reduces to $O(1)$, at the cost of an additional factor
$O(\log(n))$ in the number of catalytic qubits. This holds only in the regime of
$O(n^2/\log(n))$ catalytic qubits for Clifford unitaries, and $O(n^k/\log(n))$
catalytic qubits for diagonal gates in the $k$-th level of the Clifford
hierarchy (see Remark~\ref{rem:fanout}).

\subsection{Related work}
\paragraph{Classical catalysis.} Buhrman et al. introduced catalytic space as a model in which a computation may use an auxiliary memory containing arbitrary information, provided that the memory is restored at the end~\cite{BuhrmanCleveKouckyLoffSpeelman14}.
Subsequent work has shown that catalytic memory can increase computational power beyond ordinary low-space computation~\cite{BuhrmanCleveKouckyLoffSpeelman14,AgarwalaMertz25,ChakrabortyDattaKusreMukhopadhyaySinhababu26},
can replace other resources such as randomness and non-determinism 
\cite{CookLiMertzPyne25,KouckyMertzPyneSami25},
admits robust characterizations \cite{FolkertsmaMertzSpeelmanTupker25,KouckyMertzSami26},
and can be used to speed up low-space computation~\cite{CookPyne26,ChmelDudejaKouckyMertzRajgopal26}.
These techniques have been used in novel derandomization approaches to space complexity
\cite{LiPyneTell24,DoronPyneTellWilliams25} as well as breakthrough relationships between
time and space \cite{CookMertz25,Williams25,Shalunov26}
(for a survey on problems and techniques see \cite{Koucky16,Mertz23}).

\paragraph{Quantum catalysis.} Quantum computation has several related notions of non-clean workspace, but they impose different requirements. initialization-free quantum algorithms have been studied~\cite{Chi01,Chi05}, and Takahashi and Tani investigated uninitialized qubits in shallow circuits~\cite{TakahashiTani2021}.These works allow auxiliary qubits to begin in an unknown state, but do not formulate the requirement we need here that the entire auxiliary register be restored for every initial state. The quantum catalytic space model introduced by Buhrman et al. 
\cite{BuhrmanFolkertsmaMertzSpeelmanStrelchukSubramanianTupker25} adopts this stronger restoration requirement and develops its complexity-theoretic consequences; in particular, it gives an equivalent circuit formulation of quantum catalytic space. The usual information-theoretic notion of quantum catalysis is different again, as there the catalyst is typically a fixed state specifically chosen for the transformation, whereas our catalyst may be arbitrary and must be preserved exactly. 

Dirty ancillas have long been used as temporary workspace in quantum circuit synthesis. Numerous constructions for multi-controlled gates, arithmetic circuits, and related primitives exploit auxiliary qubits whose initial state is unknown but restored before the computation finishes. Representative techniques include toggle detection and, more recently, conditionally clean ancillas, which reduce or eliminate the overhead associated with borrowing dirty qubits \cite{Gid15,Gid17,KhattarGidney2025}. These constructions demonstrate that dirty workspace can successfully replace clean ancillas in specific circuit primitives, but do not study the global depth--workspace tradeoffs that we consider.

Catalysis in the quantum circuit model notably allows for in-place computations. This in-place catalysis is utilized in \cite{Remaud} to construct efficient algorithms for multi-controlled Tofolli gates and addition circuits.
By contrast, classical in-place computation is a provably impossible challenge for
even relatively simple functions; for example, assuming any reasonable cryptography,
neither integer multiplication nor even constant-depth constant-locality circuits
can be computed by an in-place machine~\cite{CookGhentiyalaMertzPyneSheffield25}.

Li et al. obtain $O(\log n)$-depth implementations of arbitrary CNOT circuits using $O(n^2)$ dirty ancillas, while Du, Cheng, and Ma give a catalytic compiler for bounded-fan-in, bounded-fan-out XOR circuits ~\cite{LiTianHeSun2025hammingweight,DuChengMa2026lowdepthrandomunitaries}. Their constructions demonstrate that dirty workspace can support substantial depth reductions, but do not give the depth--workspace tradeoff of Theorem~\ref{thm:catalytic-clifford}. Our construction uses a different technique: the dependence on the initial work state is cancelled algebraically by cycling the work register through a fixed family of binary linear transformations, and a separate binary matrix factorization step to remove the clean output register. We then extend the same cancellation principle to diagonal gates in the Clifford hierarchy.

\subsection{A conjecture}
The results above suggest a broader question: can catalytic workspace replace clean workspace for general quantum circuits, while preserving efficient circuit size and incurring only a small increase in depth? We formulate this as the following conjecture.

\begin{conjecture}[Space collapse]
\label{conj:collapse}
Every $n$-qubit unitary $U$ with a size $s$, depth $d$, space $a$ implementation
over a universal gate set $G$ has a corresponding catalytic implementation of size
$\poly(s)$, depth $d\cdot\mathrm{polylog}(s)$, and catalytic space $\poly(n,s,a)$.

\end{conjecture}
There are several natural weakenings of this conjecture, for example by restricting the gate set or dropping the depth bound. The results of this paper establish such a statement for Clifford and semi-Clifford circuits. Determining whether a comparable collapse holds for general quantum circuits is a novel avenue for further work that appears to require new ways of using and restoring arbitrary quantum memory.

\section{Preliminaries}

\subsection{Quantum circuits}
Our base model is the usual notion of quantum circuits.
\begin{definition}[Quantum circuit]
\label{def:circuits}
A \emph{quantum circuit} is a sequence of gates from a fixed finite gate
set $G$, with each gate acting on a bounded number of qubits. Its
\emph{size} is the number of gates, its \emph{depth} is the number of
layers, and its \emph{width} is the total number of qubits. Gates in the
same layer act on disjoint sets of qubits.
\end{definition}

We impose no restriction on which pairs of qubits may interact. All
logarithms are to base two. We write $\mathbb F_2$ for the field with two
elements and use addition over $\mathbb F_2$ for bit strings.

Our results will pertain to the Clifford circuit class, which we will next define.

\begin{definition}[Pauli group]
The single-qubit Pauli group $P_1$ consists of the operators $I,X,Y,Z$ multiplied by
phases from $\{1,-1,i,-i\}$. The $n$-qubit Pauli group $P_n$ consists of
$n$-fold tensor products of $I,X,Y,Z$ with the same four possible overall
phases, and has $4^{n+1}$ elements.
\end{definition}

\begin{definition}[Clifford circuits]
The Clifford group $C_2$ consists of the $n$-qubit unitaries $U$ satisfying
$UPU^\dagger\in P_n$ for every $P\in P_n$. Up to a global phase, it is
generated by the gate set
\[
    G_c:=\{H,S,\CNOT\},
\]
where $H$ is the Hadamard gate, $S=\operatorname{diag}(1,i)$, and $\CNOT$
is the controlled-NOT gate. A circuit over $G_c$ is called a Clifford
circuit.
\end{definition}

The Clifford hierarchy extends this condition recursively by allowing a
Pauli operator to be mapped into the preceding level under conjugation.

\begin{definition}[Clifford hierarchy]
For a fixed number $n$ of qubits, the Clifford hierarchy is the sequence
$\{C_k\}_{k\geq 1}$ defined by
\[
    C_1=P_n,\qquad
    C_k=\{U\in\mathrm{U}(2^n):UPU^\dagger\in C_{k-1}
          \text{ for every }P\in P_n\},\quad k\geq 2.
\]
\end{definition}

The first two levels are groups, while the higher levels need not be.
We use the Clifford normal form in Section~\ref{sec:clifford-background}
and the diagonal normal form in Section~\ref{sec:diagonal-hierarchy}.

\subsection{Quantum catalytic circuits}
\label{sec:catalytic-models}

We study circuits whose work qubits may start in any quantum state
and must be restored after the computation. These are also called
dirty qubits. A clean work qubit starts and ends in $\ket{0}$.
\begin{definition}[Quantum catalytic circuits]
\label{def:quantum-cat-circ}
Let $C$ be a circuit whose purpose is to perform a unitary $U$ on a set of input registers.
A \emph{catalytic qubit} may start in any state $\ket{\tau}$ and
    satisfies
    \[
        C\bigl(\ket{\psi}\ket{\tau}\bigr)
        =(U\ket{\psi})\ket{\tau}
    \]
for every $\ket{\psi}$ and $\ket{\tau}$.

The same definitions apply to registers of several qubits. For a dirty
register, the circuit must also preserve its correlations with any
external register.

A \emph{quantum catalytic circuit} $C$ for a unitary $U$ acts on an input
register $I$ and a work register $W$ of dirty qubits. Writing $I_W$ for
the identity on $W$, it satisfies
\[
    C\bigl(\ket{\psi}_I\ket{\tau}_W\bigr)
    =(U\otimes I_W)\bigl(\ket{\psi}_I\ket{\tau}_W\bigr)
\]
for every input state and every work state. Thus $C=U\otimes I_W$, and
the identity also holds for entangled states of the two registers.
\end{definition}

Note that the extension of the quantum Turing Machine definition to the
catalytic regime is equivalent to \autoref{def:quantum-cat-circ} 
\cite{BuhrmanFolkertsmaMertzSpeelmanStrelchukSubramanianTupker25}.

We sometimes identify circuits up to one global phase. In that case,
the catalytic condition is
\[
    C=e^{i\theta}(U\otimes I_W),
\]
where $\theta$ is independent of both the input and the work state.
The following observation allows us to verify this condition on
computational basis states.

\begin{lemma}[Basis state criterion]
\label{lem:globalphase_basis_state}
Suppose that a circuit $C$ satisfies
\[
    C\ket{x}_I\ket{z}_W
    =e^{i\theta}(U\ket{x}_I)\ket{z}_W
\]
for every computational basis state $\ket{x}_I\ket{z}_W$, with the
same phase $e^{i\theta}$ for all $x,z$. Then
\[
    C=e^{i\theta}(U\otimes I_W).
\]
In particular, the circuit preserves every state of the work register,
including its entanglement with other registers.
\end{lemma}

\begin{proof}
The computational basis states span the joint input and workspace, so the operator identity follows by linearity. Tensoring this identity
with the identity on any external register gives the final assertion.
\end{proof}

Returning each work basis state up to a phase is not sufficient.
For example, $Z$ preserves each computational basis state up to a
phase but changes $\ket{+}$ to $\ket{-}$. The phase must be the same
for every work basis state.

When a circuit is applied conditionally, its overall phase becomes
a relative phase. Throughout the constructions that follow, we specify whether the circuit equals $U\otimes I_W$ exactly or only up to one
global phase.

\section{Technical overview}

A $\CNOT$ circuit acts as an invertible binary linear map on computational
basis states. For an $n\times n$ binary matrix $M$ and $n$-bit strings
$x,y$, the associated addition
\[
    \ket{x}\ket{y}\longmapsto\ket{x}\ket{y+Mx}
\]
is reversible even when $M$ is singular. This addition is the main step
in our construction. The parallel circuit of Jiang et al. restores its
work register, but the output can depend on the initial work string.
We cancel this dependence by combining three copies of their circuit
with changes of basis on the work register, as shown in
Lemma~\ref{lem:roots-of-unity-trick}.

To implement an invertible map on the input itself,
we decompose its
matrix into six additions of this form. Each addition uses one half of
the input as its control register and the other half as its target
register, so no separate clean output register is needed. The resulting
$\CNOT$ construction extends to Clifford circuits through the normal
form of Aaronson and Gottesman~\cite{AaronsonGottesman2004}. In particular,
Lemma~\ref{cor:catalytic-clifford} gives depth $O(\log n)$ with
$O(n^2/\log^2 n)$ catalytic qubits and no clean qubits.

For diagonal elements of the Clifford hierarchy, we must further cancel the dependence
of a phase on the work state. The normal form writes the required phase
as a sum of monomials of bounded degree and precision. For each monomial,
we compare the phase before and after adding the input to a work block,
then repeat this comparison over a cycle of binary linear maps. The
accumulated phase depends only on the input. Separate work blocks allow
the monomials to be implemented in parallel, with an additional logarithmic
depth to distribute shared input controls. Section~\ref{sec:diagonal-hierarchy}
gives the construction and its extension to semi-Clifford gates.

\section{Clifford circuits}
\label{sec:clifford-circuits}

We now give the Clifford construction outlined above, allowing every
catalytic qubit to start in an arbitrary state. We first recall the
Clifford normal form and the role of clean auxiliary qubits in the Jiang
construction, then show how to remove the initialisation requirements
from the work and output registers. Clifford unitaries are identified
up to one global phase throughout this section.

\subsection{Background}
\label{sec:clifford-background}

The normal form of Aaronson and Gottesman reduces Clifford synthesis
to a constant number of $\CNOT$ circuits separated by single-qubit gates.

\begin{theorem}[Clifford normal form~{\cite[Theorem~8]{AaronsonGottesman2004}}]
\label{thm:aaronson-gottesman}
Every Clifford circuit over $G_c$ can be converted
into an equivalent circuit with eleven rounds:
\[
    H-C-S-C-S-C-H-S-C-S-C,
\]
where each $H$ is a layer of Hadamard gates, each $S$ is a layer of $S$ gates,
and each $C$ is a $\CNOT$ circuit.
\end{theorem}

Since the single-qubit layers have constant depth, it suffices to
parallelise the $\CNOT$ circuits. With clean auxiliary qubits, Jiang et
al.~\cite{doi:10.1137/1.9781611975994.13} obtained the following optimal
relation between depth and workspace.

\begin{theorem}[{\cite[Theorem~4.1]{doi:10.1137/1.9781611975994.13}}]
\label{thm:cnot-shrink}
    Let $1\leq s\leq O(n/\log^2 n)$. Every $n$-qubit CNOT circuit can be transformed
    into an equivalent CNOT circuit of depth
    \[
        O\!\left(\frac{n}{s\log n}\right)
    \]
    that uses $(3s+1)n$ ancillas.
\end{theorem}

They also provide a matching lower bound to proof that their result is optimal:

\begin{theorem}[{\cite[Theorem~1.3]{doi:10.1137/1.9781611975994.13}}, with $m=sn$]
\label{thm:jiang-lowerbound}
    Let $0\leq\epsilon<\sqrt2/2$ be a constant and let $0<s<n/\log^2 n$. For a $1-o(1)$
    fraction of $n$-qubit CNOT circuits, every $\epsilon$-approximate $n$-qubit quantum
    circuit with $sn$ ancillas has depth
    \[
        \Omega\!\left(\frac{n}{(1+s)\log n}\right).
    \]
\end{theorem}

Applying this result to the constant number of $\CNOT$ circuits in
Theorem~\ref{thm:aaronson-gottesman} gives the same bound for Clifford
circuits.

\begin{corollary}[of {\cite[Theorem~4.1 and Corollary~1.2]{doi:10.1137/1.9781611975994.13}}]
\label{cor:jiang-clifford-s}
    Let $1\leq s\leq O(n/\log^2 n)$. Every $n$-qubit Clifford circuit can be implemented,
    without measurements, by a circuit of depth
    \[
        O\!\left(\frac{n}{s\log n}\right)
    \]
    that uses $(3s+1)n$ ancillas initialized to $\ket{0}$ and returns them to $\ket{0}$.
    This is optimal in the worst case, since CNOT circuits are Clifford circuits and
    Corollary~\ref{thm:jiang-lowerbound} applies to them.
\end{corollary}

Our construction achieves logarithmic depth using
$O(n^2/\log^2 n)$ catalytic qubits, showing that this degree of parallelization in fact
does not require clean workspace. To explain the adaptation, we first
describe $\CNOT$ circuits as binary linear maps.

For bits $x,y$, a $\CNOT$ gate acts as
\[
    \CNOT\ket{x}\ket{y}=\ket{x}\ket{y+x},
\]
which is multiplication of the column vector $(x,y)^{\mathsf T}$ by
\[
    \begin{bmatrix}
        1&0\\
        1&1
    \end{bmatrix}.
\]
We write $\operatorname{GL}(n,2)$ for the group of invertible $n\times n$
binary matrices. Every $n$-qubit $\CNOT$ circuit defines a matrix in this
group, and every such matrix can be reduced to the identity by row
elimination. Since a row addition is a $\CNOT$ gate and a row swap is a
product of three row additions, each $M\in\operatorname{GL}(n,2)$ has a
$\CNOT$ circuit $\mathcal C$ satisfying
\[
    \mathcal C\ket{x}=\ket{Mx}
\]
for every $x\in\mathbb F_2^n$~\cite{PatelMarkovHayes2008}. Matrix and
circuit products act from right to left.

The usual implementation with a clean output register uses reversible
addition. Let $I$ and $O$ be $n$-qubit input and output registers, and let
$\mathcal C_1$ implement the block matrix
\[
    \begin{bmatrix}
        I&0\\
        M&I
    \end{bmatrix},
\]
so that
\[
    \mathcal C_1\ket{x}_I\ket{y}_O
    =\ket{x}_I\ket{y+Mx}_O,
\]
where the identity blocks act on the corresponding registers. A second
circuit $\mathcal C_2$ implements
\[
    \begin{bmatrix}
        I&0\\
        M^{-1}&I
    \end{bmatrix}.
\]
Starting with $O$ in $\ket{0}$, the sequence of $\mathcal C_1$, a swap of
the two registers, and $\mathcal C_2$ computes $Mx$ in $I$ and restores
$O$ to $\ket{0}$:
\begin{align*}
    \mathcal C_1\ket{x}_I\ket{0}_O
        &=\ket{x}_I\ket{Mx}_O,\\
    \SWAP\ket{x}_I\ket{Mx}_O
        &=\ket{Mx}_I\ket{x}_O,\\
    \mathcal C_2\ket{Mx}_I\ket{x}_O
        &=\ket{Mx}_I\ket{x+M^{-1}Mx}_O\\
        &=\ket{Mx}_I\ket{0}_O.
\end{align*}

Additional work qubits allow parts of these additions to run in parallel.
Our next lemma describes the resulting action when the work register is
not assumed to start in $\ket{0}$.

\subsection{Catalytic implementation of \autoref{thm:cnot-shrink}}
\label{sec:catalytic-linear}

Lemma~4.2 of Jiang et al.~\cite{doi:10.1137/1.9781611975994.13} gives
the required parallel addition circuit. Although stated for an invertible
off-diagonal block, its proof also applies to arbitrary binary matrices.
We use the numbering in the full version of their paper.

\begin{lemma}[Adapted from Lemma~4.2 of~\cite{doi:10.1137/1.9781611975994.13}]
\label{lem:clean-cnot-parallelisation}
Let $n\geq 2$ and let $s$ be an integer with $1\leq s\leq n/\log^2 n$.
For every $M\in\mathbb F_2^{n\times n}$, there exists a $\CNOT$ circuit
$\mathcal C_M$ of depth
\[
    O\!\left(\frac{n}{s\log n}\right)
\]
using a work register $W$ of $m=O(sn)$ qubits. Its input and output
registers $I,O$ each contain $n$ qubits, and for some
$B\in\mathbb F_2^{n\times m}$ it implements the matrix
\[
    C_M=
    \begin{bmatrix}
        I_I&0&0\\
        M&I_O&B\\
        0&0&I_W
    \end{bmatrix}.
\]
Equivalently,
\[
    \mathcal C_M\ket{x}_I\ket{y}_O\ket{z}_W
    =\ket{x}_I\ket{y+Mx+Bz}_O\ket{z}_W
\]
for $x,y\in\mathbb F_2^n$ and $z\in\mathbb F_2^m$. Subscripts on
identity blocks indicate their registers.
\end{lemma}

\begin{proof}
Although Jiang et al. state the lemma for invertible $M$, their proof
constructs $M$ one block at a time by partitioning it into groups of
$s\log^2 n$ columns and applying their Corollary~4.5 to each group.
That corollary allows an arbitrary binary block, so the construction and
its resource bounds apply unchanged to any $M\in\mathbb F_2^{n\times n}$. Jiang et al. do not specify the blocks acting on $W$, since their workspace starts in $\ket{0}$. Their circuit uses $I$ only as a control, so the $(I,W)$ block is $0$. It restores $W$ by applying the inverse of the circuit that computed into it, so the $(W,W)$ block is $I_W$.
\end{proof}

\begin{remark}
\label{rem:fanout}
The construction of Jiang et al.~\cite{doi:10.1137/1.9781611975994.13} can be improved in depth
if one has access to $O(n^2/\log(n))$ auxiliary qubits and the fan-out gate
(one control, many NOT targets) at unit cost. Then the remaining $O(\log(n))$ depth
reduces to depth $O(1)$. The implementation then consists of copying the data $x$
to many registers, or of computing the parity gate, and each of these happens only
a constant number of times. Utilizing fan-out and parity natively comes at the cost
of an additional multiplicative factor $O(\log(n))$ in the number of auxiliary
qubits. The improvement in depth therefore only holds when one has access to
$O(n^2/\log(n))$ auxiliary qubits. All our constructions incur depth scaling in $n$
either through the Jiang et al.\ construction or through fan-out used to copy
information. Therefore, this improvement in depth also holds for our constructions. This might be especially interesting for hardware implementations with native
multi-target interactions. For example, the proposed implementations of fan-out using collective interactions in trapped ions~\cite{MaslovNam2018}, cavity systems~\cite{YangLiuNori2010}, and Rydberg atoms~\cite{MullerEtAl2009}. These implementations use a number of collective operations independent of the number of targets.

\end{remark}

With both the output and work registers initially zero, this circuit
computes $Mx$ in the output register:
\[
    \mathcal C_M\ket{x}_I\ket{0}_O\ket{0}_W
    =\ket{x}_I\ket{Mx}_O\ket{0}_W.
\]
For an arbitrary work string $z$, the output instead contains the
additional term $Bz$. We eliminate this term by cycling the work register
through three binary linear maps whose sum is zero.

The construction uses
\[
    P:=
    \begin{pmatrix}
        0&1\\
        1&1
    \end{pmatrix}.
\]
\begin{samepage}
This matrix maps $(z_1,z_2)^{\mathsf T}$ to
$(z_2,z_1+z_2)^{\mathsf T}$ and has a two-$\CNOT$ implementation:
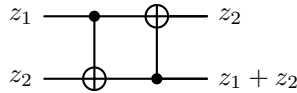
\begin{center}
    \begin{quantikz}
        \lstick{$z_1$} & \ctrl{1} & \targ{} & \rstick{$z_2$} \qw \\
        \lstick{$z_2$} & \targ{} & \ctrl{-1} & \rstick{$z_1+z_2$} \qw
    \end{quantikz}
\end{center}
\end{samepage}

Its powers satisfy
\[
    P^3=I
    \qquad\text{and}\qquad
    I+P+P^2=0
\]
over $\mathbb F_2$. Applying these maps to pairs of work qubits lets us
cancel $Bz$ while restoring the entire work register.

\begin{lemma}[Catalytic binary addition]
\label{lem:roots-of-unity-trick}
Let $n\geq 2$ and let $s$ be an integer with $1\leq s\leq n/\log^2 n$.
For every $M\in\mathbb F_2^{n\times n}$, there exists a $\CNOT$ circuit
$\widetilde{\mathcal C}_M$ of depth
\[
    O\!\left(\frac{n}{s\log n}\right)
\]
using $O(sn)$ catalytic work qubits that implements
\[
    \widetilde C_M=
    \begin{bmatrix}
        I_I&0&0\\
        M&I_O&0\\
        0&0&I_W
    \end{bmatrix}.
\]
Equivalently, for every $x,y\in\mathbb F_2^n$ and every work state
$\ket{\tau}_W$,
\[
    \widetilde{\mathcal C}_M
    \ket{x}_I\ket{y}_O\ket{\tau}_W
    =\ket{x}_I\ket{y+Mx}_O\ket{\tau}_W.
\]
\end{lemma}

\begin{proof}
Let $m$ be the number of work qubits in the circuit of
Lemma~\ref{lem:clean-cnot-parallelisation}. We first treat even $m$ and
combine three conjugated copies of its matrix $C_M$. Define
\[
\mathbf P
:=
\underbrace{
\begin{pmatrix}
P & 0 & \cdots & 0\\
0 & P & \ddots & \vdots\\
\vdots & \ddots & \ddots & 0\\
0 & \cdots & 0 & P
\end{pmatrix}
}_{m/2\text{ copies}}
\qquad\text{and}\qquad
D:=
\begin{pmatrix}
I_I & 0 & 0\\
0 & I_O & 0\\
0 & 0 & \mathbf P
\end{pmatrix}.
\]
Thus, $\mathbf P$ applies $P$ to each consecutive pair of work qubits. The
matrix $D$ can be implemented in depth two using two $\CNOT$ gates per pair.

For $j\in\{0,1,2\}$, define
\[
    C_{M,j}:=D^{-j}C_MD^j
    =
    \begin{bmatrix}
        I&0&0\\
        M&I&B\mathbf P^j\\
        0&0&I
    \end{bmatrix}.
\]
Because $P^3=I$, we have $P^{-1}=P^2$. Each $C_{M,j}$ therefore consists of
one application of $\mathcal C_M$ and constant-depth $\CNOT$ circuits on
$W$ before and after it. Finally,
\begin{align*}
    C_{M,0}C_{M,1}C_{M,2}
    &=
    \begin{bmatrix}
        I&0&0\\
        M+M+M&I&B+B\mathbf P+B\mathbf P^2\\
        0&0&I
    \end{bmatrix}\\
    &=
    \begin{bmatrix}
        I_I&0&0\\
        M&I_O&0\\
        0&0&I_W
    \end{bmatrix}
    =\widetilde C_M.
\end{align*}
Here we used $M+M+M=M$ and
$I+\mathbf P+\mathbf P^2=0$ over $\mathbb F_2$.

If $m$ is odd, append one additional catalytic qubit and extend $B$ by
one zero column. The enlarged work register has even size, and the same
argument applies. The additional qubit does not change the asymptotic space
bound.
\end{proof}

The matrix identity holds for every computational basis state of the work
register. By linearity, it holds for every quantum state, including states
entangled with an external register. In particular, the circuit acts exactly
as the identity on the catalytic register.

\subsection{Removing the clean output register}
\label{sec:in-place-cnot}

The preceding lemma implements addition into an arbitrary output
register. To use it for a transformation of the input itself, we apply
the following special case of Urschel's
factorisation~\cite{urschel2023}.

\begin{theorem}[Urschel~{\cite[Theorem~1]{urschel2023}}]
\label{thm:urschel}
For every $M\in\operatorname{GL}(2n,2)$, there exist binary matrices
$A_1,\ldots,A_6\in\mathbb F_2^{n\times n}$ such that
\[
    M=
    \begin{pmatrix}I&0\\A_1&I\end{pmatrix}
    \begin{pmatrix}I&A_2\\0&I\end{pmatrix}
    \begin{pmatrix}I&0\\A_3&I\end{pmatrix}
    \begin{pmatrix}I&A_4\\0&I\end{pmatrix}
    \begin{pmatrix}I&0\\A_5&I\end{pmatrix}
    \begin{pmatrix}I&A_6\\0&I\end{pmatrix}.
\]
\end{theorem}

Urschel proves the factorisation for matrices of determinant one over
an arbitrary field. Since every invertible binary matrix has determinant
one, it applies in the form above. Each factor is an addition from one
half of the input to the other, which we implement using
Lemma~\ref{lem:roots-of-unity-trick}.

\begin{theorem}
\label{thm:catalytic-cnot}
Let $n\geq 2$ and let $s$ be an integer with $1\leq s\leq n/\log^2 n$.
For every $M\in\operatorname{GL}(n,2)$, there exists a catalytic
$\CNOT$ circuit $\mathcal C_c$ of depth
\[
    O\!\left(\frac{n}{s\log n}\right)
\]
using $O(sn)$ catalytic qubits. For every $x\in\mathbb F_2^n$ and every
work state $\ket{\tau}_W$,
\[
    \mathcal C_c
    \bigl(\ket{x}_I\ket{\tau}_W\bigr)
    =
    \ket{Mx}_I\ket{\tau}_W.
\]
\end{theorem}

\begin{proof}
We construct the circuit on computational basis states and then use
linearity. It suffices to treat sufficiently large $n$, since bounded
values can be handled by Gaussian elimination without additional qubits.
We may also take $n=2r$ to be even: for odd $n$, append one catalytic
qubit and apply the even-dimensional construction to $\begin{pmatrix}
M & 0\\
0 & 1
\end{pmatrix}.
$
The extended map leaves this qubit unchanged and has the same asymptotic
resource bounds.
By Theorem~\ref{thm:urschel}, there exist
$A_1,\ldots,A_6\in\mathbb F_2^{r\times r}$ such that
\[
    M=M_1M_2M_3M_4M_5M_6,
\]
where
\[
    M_i=
    \begin{cases}
        \begin{pmatrix}I&0\\A_i&I\end{pmatrix},&i\text{ odd},\\[1ex]
        \begin{pmatrix}I&A_i\\0&I\end{pmatrix},&i\text{ even}.
    \end{cases}
\]

Split the input into two $r$-qubit halves and use the first half as the
control register for odd $i$, reversing the roles for even $i$. To apply
Lemma~\ref{lem:roots-of-unity-trick} in dimension $r$, choose
$s':=\min\{s,\lfloor r/\log^2 r\rfloor\}$. For sufficiently large $n$,
$s'\geq 1$ and $s'=\Theta(s)$, so this choice preserves both asymptotic
bounds. The lemma gives a circuit $\mathcal C_i$ for each $M_i$.
For odd $i$, call the two halves $I$ and $O$ and write their basis
labels as $x_{\mathrm{in}}$ and $x_{\mathrm{out}}$. The circuit acts as
\[
    \mathcal C_i
    \ket{x_{\mathrm{in}}}_I
    \ket{x_{\mathrm{out}}}_O
    \ket{z}_W
    =
    \ket{x_{\mathrm{in}}}_I
    \ket{x_{\mathrm{out}}+A_ix_{\mathrm{in}}}_O
    \ket{z}_W.
\]
The analogous identity holds for even $i$ with the two data registers
interchanged.

The circuit for $M$ is their product
\[
    \mathcal C_c
    :=\mathcal C_1\mathcal C_2\mathcal C_3
      \mathcal C_4\mathcal C_5\mathcal C_6,
\]
with the rightmost factor acting first. For every computational basis
state $\ket{x}$ and work basis state $\ket{z}$, it satisfies
\[
    \mathcal C_c\ket{x}_I\ket{z}_W
    =\ket{Mx}_I\ket{z}_W.
\]
No phase depends on $x$ or $z$, since $\mathcal C_c$ consists only of
$\CNOT$ gates. Lemma~\ref{lem:globalphase_basis_state} extends the
identity to arbitrary states of the input and work registers.

Each $\mathcal C_i$ uses $O(s'r)=O(sn)$ catalytic qubits and has depth
$O(r/(s'\log r))=O(n/(s\log n))$. Since every factor restores the work
register, all six circuits can reuse it, giving total workspace $O(sn)$
and depth $O(n/(s\log n))$.
\end{proof}

This result is in fact optimal. The lower bound given in Theorem~\ref{thm:jiang-lowerbound} by Jiang et. al. \cite{doi:10.1137/1.9781611975994.13} holds for clean qubits, therefore in particular, it also holds for catalytic qubits.
\subsection{From \texorpdfstring{$\CNOT$}{CNOT} circuits to Clifford circuits}

The Clifford normal form contains only a constant number of $\CNOT$
circuits. Applying Theorem~\ref{thm:catalytic-cnot} to each of them
therefore gives the same asymptotic depth and workspace for Clifford
circuits, which will be used to prove \autoref{thm:catalytic-clifford}.

\begin{lemma}
\label{cor:catalytic-clifford}
Let $n\geq 2$ and let $s$ be an integer with $1\leq s\leq n/\log^2 n$.
For every $n$-qubit Clifford circuit $\mathcal C$, there exists a
Clifford circuit $\mathcal C_c$ of depth
\[
    O\!\left(\frac{n}{s\log n}\right)
\]
using $O(sn)$ catalytic qubits, such that, for every input state $\ket{\psi}$ and every
work state $\ket{\tau}_W$,
\[
    \mathcal C_c
    \bigl(\ket{\psi}_I\ket{\tau}_W\bigr)
    =
    (\mathcal C\otimes I_W)
    \bigl(\ket{\psi}_I\ket{\tau}_W\bigr),
\]
up to the global phase convention for Clifford circuits.
\end{lemma}

\begin{proof}
Use the eleven rounds of Theorem~\ref{thm:aaronson-gottesman} and replace
each $\CNOT$ circuit by its catalytic implementation from
Theorem~\ref{thm:catalytic-cnot}. These circuits restore the work
register after every round, so they may all use the same catalytic
qubits. The intervening single-qubit gates add only constant depth,
which gives the claimed bounds.
\end{proof}

\subsection{Application to constructing low-T count approximate Toffoli}
\label{sec:Toffoli_construction}
An important subroutine in many quantum algorithms is the multi-controlled Toffoli gate,
which computes $\mathsf{AND}_n$ of $n$ control qubits into a target qubit. For instance,
up to Clifford gates, this gate forms the diffusion operator used in Grover's algorithm.
It has been shown that an exact implementation of this gate requires at least $n$
T gates, or magic states, even when measurements and classical feed-forward are
allowed~\cite{Beverland_2020}. 

Gosset et al.~\cite[Theorem~1]{GKZ25} showed that this $\mathsf{T}$-count can be
exponentially improved when allowing some error on the implementation of the Toffoli gate:
$O(\log(1/\epsilon))$ $\mathsf{T}$ gates suffice for every $n$ and $\epsilon\in(0,1/2]$. This
result holds specifically in the mixed model of unitary implementation. Mixed here means that
they generate an ensemble of unitary circuits, of which one is picked at random and then
implemented. For every computational basis input, the chosen circuit acts correctly with
probability at least $1-\epsilon$, and the averaged channel is $\epsilon$-close to the Toffoli
gate in diamond norm~\cite[Definition~10]{GKZ25}. The $\mathsf{T}$-count bound holds for all
circuits in the ensemble.

In this section we show that their construction can be improved in depth when given access to
catalytic qubits, and be made optimal when enough are available. This could be of use in
general compilers and reduce the overall overhead in quantum circuits. The depth reduction
scales with the number of catalytic qubits available. The direct implementation of the
construction of~\cite{GKZ25} has depth $O(\log(n)\log(1/\epsilon))$. An algorithm that uses
this gate $\mathrm{poly}(n)$ times needs per-gate error $\epsilon=1/\mathrm{poly}(n)$ for a
constant total error, by the union bound. The depth per gate then becomes $O(\log^2(n))$. We
reduce this to $O\bigl(\log(n)\frac{\log(1/\epsilon)}{s}+\log\log(1/\epsilon)\bigr)$ when
having access to $sn$ catalytic qubits. This reduces to $O(\log(n))$ depth when
$s=\Theta(\log(n))$, which is optimal for the $n$-qubit Toffoli gate among circuits of
two-qubit gates.

\begin{lemma}[{Random-parity $\mathsf{OR}$, based on~\cite[Algorithm~1 and Theorem~12]{GKZ25}}]
\label{lem:random-or}
For any positive integer $n$ and $\epsilon < 1/2$,
there exists an ensemble of circuits $\{C_j\}_j$ of depth $O(\log(n)\log(1/\epsilon))$ using
$O(\log(1/\epsilon))$ T gates and $O(\log(1/\epsilon))$ clean work space, such that every
$C_j$ returns $W$ to $\ket{0}$ and, for every $x \in \{0,1\}^n$ and $b\in\{0,1\}$, with
probability at least $1-\epsilon$ over a uniformly random $j$,
\[
    C_j\ket{x}_{I}\ket{b}_{O}\ket{0}_W = \ket{x}_I\ket{b\oplus\mathsf{OR}_n(x)}_O\ket{0}_W.
\]
Moreover, the channel
\[
    \mathcal{E}(\rho)=\mathbb{E}_j\bigl[C_j\,(\rho\otimes\ket{0}\!\bra{0}_W)\,C_j^\dagger\bigr]
\]
equals $\mathcal{E}'(\rho)\otimes\ket{0}\!\bra{0}_W$ with
$D_\diamond\bigl(\mathcal{E}',U_{\mathsf{OR}_n}\bigr)\le\epsilon$.
\end{lemma}

\begin{proof}
For simplicity, start by setting $k = \lceil \log(1/\epsilon) \rceil + 2$, so that $2^{-k} \leq \epsilon/4$; the extra 2 will be needed only for the diamond-distance bound at the end. 
Instead of implementing an exact $\mathsf{OR}_n$ gate, the circuit implements $k$
$\mathsf{Parity}_S$ gates on uniformly random subsets $S$.  Start by picking a collection of
subsets $\{S_i\}_{i=1}^{k}$, each $S_i \subseteq [n]$ chosen uniformly at random and
independently. If $\mathsf{OR}_n(x) = 0$, then $\mathsf{Parity}_{S_i}(x) = 0$ for every
$i \in [k]$, but if $\mathsf{OR}_n(x) = 1$, then for every $i$, $\mathsf{Parity}_{S_i}(x) = 1$
with probability $1/2$, independently. Let $C_j'$ be the circuit that sequentially applies
$\mathsf{Parity}_{S_i}$ to $x$, with a separate output qubit in $W$ for every $i$:
\[
    C_j'\ket{x}_I\ket{b}_O\ket{0}_W = \ket{x}_I\ket{b}_O\bigotimes_{i=1}^{k}\ket{\mathsf{Parity}_{S_i}(x)}_W .
\]
Let $C_j$ now be the circuit that, after $C_j'$, applies the exact $\mathsf{OR}_{k}$ gate with
the work register as input and the output register as output. Then apply $C_j'^{\dagger}$ to
clean the work register. Every $C_j$ therefore acts as
$C_j\ket{x}_I\ket{b}_O\ket{0}_W=\ket{x}_I\ket{b\oplus g_j(x)}_O\ket{0}_W$ with
$g_j=\mathsf{OR}_k(\mathsf{Parity}_{S_1},\dots,\mathsf{Parity}_{S_k})$. If $x = 0$, thus
$\mathsf{OR}_n(x) = 0$, then $g_j(x)=0$ for every $j$, as required. If
$\mathsf{OR}_n(x) = 1$, then the probability that all $\mathsf{Parity}_{S_i}(x) = 0$ is
$\left(\tfrac{1}{2}\right)^{k} \le \epsilon/4$. Therefore, with probability at least
$1-\epsilon$, one of the $\mathsf{Parity}_{S_i}(x) = 1$ and $g_j(x)=1=\mathsf{OR}_n(x)$.

$C_j$ is the circuit for one choice of the collection $\{S_i\}$. Every $j$ corresponds to
another collection $\{S_i\}$, and the same analysis holds for all of them. We can count the
complexity of $C_j$. There are $k$ $\mathsf{Parity}_{S_i}$ gates, each a CNOT tree of depth
$O(\log n)$, and one exact $\mathsf{OR}_{k}$ gate, a tree of Toffoli gates of depth
$O(\log k)$ with $O(k)$ T gates. Remeber that $k = \lceil \log(1/\epsilon) \rceil + 2$, therefore, the depth is
$O(\log(n)\log(1/\epsilon) + \log\log(1/\epsilon)) = O(\log(n)\log(1/\epsilon))$, it uses
$O(\log(1/\epsilon))$ $T$ gates, and it uses $O(\log(1/\epsilon))$ additional clean qubits.

For the diamond distance, let $U_j\ket{x}\ket{b}=\ket{x}\ket{b\oplus g_j(x)}$. By linearity,
$C_j(\ket{\psi}\otimes\ket{0}_W)=(U_j\ket{\psi})\otimes\ket{0}_W$ for every state
$\ket{\psi}$ on $I$, $O$ and a reference system, so $\mathcal{E}$ factors with
$\mathcal{E}'(\rho)=\mathbb{E}_j[U_j\rho\,U_j^\dagger]$. This is the channel
of~\cite[Algorithm~1]{GKZ25} up to fixed $X$ gates, since
$U_{\mathsf{OR}_n}=(X^{\otimes n}\otimes I)\,\mathrm{Toff}_{n+1}\,X^{\otimes n+1}$ and the same
gates relate $U_j$ to their $W_g$. Fixed unitaries before and after preserve the diamond
distance, so~\cite[Theorem~12]{GKZ25} gives
$D_\diamond(\mathcal{E}',U_{\mathsf{OR}_n})\le 4\cdot 2^{-k}\le\epsilon$.
\end{proof}

The depth of this ensemble of circuits is dominated by the $k$ $\mathsf{Parity}_{S_i}$ gates.
The $\mathsf{OR}_k$ gate adds only depth $O(\log k)=O(\log\log(1/\epsilon))$, so
\[
    \operatorname{depth}(C_j)=2\operatorname{depth}(C'_j)+O(\log\log(1/\epsilon)).
\]

The circuit $C'_j$ is a CNOT circuit: it computes the linear map
$\ket{x}_I\ket{y}_W\mapsto\ket{x}_I\ket{y\oplus Mx}_W$, where the rows of
$M\in\F_2^{k\times n}$ are the indicator vectors of $S_1,\dots,S_k$. Linear maps of this form
are exactly what the catalytic compiler parallelizes. We use it to give a depth-improved
implementation of the approximate $\mathsf{OR}_n$ gate.
\begin{theorem}\label{thm:catalytic-or}
For any integer $n\ge 2$, any $\epsilon$ with $2^{-n}\le\epsilon<1/2$ and any integer $s$ with
$1\le s\le\lceil\log(1/\epsilon)\rceil$, there exists an ensemble of catalytic circuits $\{C_j\}_j$ of
depth $O\bigl(\log(n)\frac{\log(1/\epsilon)}{s}+\log\log(1/\epsilon)\bigr)$ using $sn$
catalytic qubits, $O(\log(1/\epsilon))$ T gates and $O(\log(1/\epsilon))$ clean qubits, such
that the channel
\[
    \mathcal{E}(\rho\otimes\ket{0}\!\bra{0}_{W_c}\otimes\ket{\tau}\bra{\tau}_{W_{cat}})
    =\mathbb{E}_j\bigl[C_j\,(\rho\otimes\ket{0}\!\bra{0}_{W_c}\otimes\ket{\tau}\bra{\tau}_{W_{cat}})\,C_j^\dagger\bigr]
\]
returns the registers $W_c$ and $W_{cat}$ exactly for every catalyst state $\ket{\tau}$,
\[
    \mathcal{E}(\rho\otimes\ket{0}\!\bra{0}_{W_c}\otimes\tau_{W_{cat}})
    =\mathcal{E}'(\rho)\otimes\ket{0}\!\bra{0}_{W_c}\otimes\tau_{W_{cat}},
\]
and is close to the $\mathsf{OR}_n$ gate in diamond norm,
$D_\diamond\bigl(\mathcal{E}',U_{\mathsf{OR}_n}\bigr)\le\epsilon$.
\end{theorem}

\begin{proof}
In principle, we use the same construction as in Lemma~\ref{lem:random-or}, with the
adjustment that the circuit $C'_j$ consisting of the $\mathsf{Parity}_{S_i}$ gates is
implemented using the catalytic compiler. We prove this first for the catalyst in a
computational basis state $\ket{z}$ and then invoke Lemma~\ref{lem:globalphase_basis_state}
to prove it for general $\tau$. Let $k=\lceil\log(1/\epsilon)\rceil+2$, and split the catalyst
into $s$ rows $W_{cat}=Z_1\cdots Z_s$ of $n$ qubits each, holding $z_1,\dots,z_s$. Batch the
$k$ parity gates into $\lceil k/s\rceil$ batches of at most $s$ gates. For a batch with
subsets $S_1,\dots,S_s$ (after relabelling) and target qubits $a_1,\dots,a_s$ in $W_c$,
perform the following circuit. Use fanout from Lemma~\ref{lem:dirty-target-fanout} to copy
$x$ into the catalyst:
\[
    F\ket{x}\bigotimes_{l=1}^{s}\ket{0}_{a_l}\ket{z_l}_{Z_l}
    \longmapsto\ket{x}\bigotimes_{l=1}^{s}\ket{0}_{a_l}\ket{z_l\oplus x}_{Z_l}.
\]
Then, apply the batch of $\mathsf{Parity}_{S_l}$ gates in parallel, the $l$-th one reading
row $Z_l$ and writing into $a_l$:
\[
    \bigotimes_{l=1}^{s}U_{\mathsf{Parity}_{S_l}}\,
    \ket{x}\bigotimes_{l=1}^{s}\ket{0}_{a_l}\ket{z_l\oplus x}_{Z_l}
    \longmapsto\ket{x}\bigotimes_{l=1}^{s}\ket{\mathsf{Parity}_{S_l}(x\oplus z_l)}_{a_l}\ket{z_l\oplus x}_{Z_l}.
\]
Each $U_{\mathsf{Parity}_{S_l}}$ is an in-place CNOT tree on the qubits of $Z_l$ in $S_l$, one
CNOT into $a_l$, and the inverse tree. It has depth $O(\log n)$ and leaves $Z_l$ unchanged,
and the rows are disjoint, so the batch runs in parallel. Remove $x$ from the catalytic
register using fanout again:
\[
    F\ket{x}\bigotimes_{l=1}^{s}\ket{\mathsf{Parity}_{S_l}(x\oplus z_l)}_{a_l}\ket{z_l\oplus x}_{Z_l}
    \longmapsto\ket{x}\bigotimes_{l=1}^{s}\ket{\mathsf{Parity}_{S_l}(x\oplus z_l)}_{a_l}\ket{z_l}_{Z_l}.
\]
Finally, apply the batch of $\mathsf{Parity}_{S_l}$ gates one more time:
\[
    \bigotimes_{l=1}^{s}U_{\mathsf{Parity}_{S_l}}\,
    \ket{x}\bigotimes_{l=1}^{s}\ket{\mathsf{Parity}_{S_l}(x\oplus z_l)}_{a_l}\ket{z_l}_{Z_l}
    \longmapsto\ket{x}\bigotimes_{l=1}^{s}\ket{\mathsf{Parity}_{S_l}(x)}_{a_l}\ket{z_l}_{Z_l},
\]
where we use the fact that
$\mathsf{Parity}_{S}(x\oplus z)\oplus\mathsf{Parity}_{S}(z)=\mathsf{Parity}_{S}(x)$. Repeat
this for every batch, reusing the same rows $Z_l$; this implements $C'_j$ and returns the
catalyst to $\ket{z}$. The rest of the circuit is the same as in Lemma~\ref{lem:random-or}.

Because the catalytic implementation of $C'_j$ has exactly the same action as the sequential
one on every basis state, and fixes $z$, each $C_j$ acts as the circuit of
Lemma~\ref{lem:random-or} on $I$, $O$ and $W_c$, and as the identity on $W_{cat}$. By
Lemma~\ref{lem:globalphase_basis_state} this holds for every catalyst state $\tau$, also when
$W_{cat}$ is entangled with a reference system. Hence $\mathcal{E}$ factors as stated, and
the diamond-norm bound of Lemma~\ref{lem:random-or} carries over, with $W_{cat}$ taken as part
of the reference system.

Accounting for complexity: one batch has depth $O(\log n+\log s)$, from two fanouts of depth
$O(\log s)$ and two parity layers of depth $O(\log n)$. Since $\epsilon\ge 2^{-n}$, we have
$s\le\log(1/\epsilon)\le n$, so $\log s\le\log n$. There are
$\lceil k/s\rceil\le 2k/s=O\bigl(\frac{\log(1/\epsilon)}{s}\bigr)$ batches, so the catalytic
$C'_j$ has depth $O\bigl(\log(n)\frac{\log(1/\epsilon)}{s}\bigr)$ and uses $sn$ catalytic
qubits. The $\mathsf{OR}_k$ gate has depth $O(\log k)=O(1+\log\log(1/\epsilon))$ and uses
$O(\log(1/\epsilon))$ T gates. The final depth is therefore
$O\bigl(\log(n)\frac{\log(1/\epsilon)}{s}+\log\log(1/\epsilon)\bigr)$, where the constant is
absorbed by the first term.
\end{proof}

\section{Generalizations}
\subsection{Diagonal elements of the Clifford hierarchy}
\label{sec:diagonal-hierarchy}

The cancellation used for binary addition has an analogue for diagonal
phases. We develop it for $n$-qubit diagonal unitaries in a fixed level
$C_k$, with $k\geq 2$, where the level bounds both the degree and the
precision of the phase polynomial. As before, $W$ denotes the catalytic
work register and $I_W$ its identity operator.

For $x\in\{0,1\}^n$ and a nonempty set $S\subseteq\{1,\ldots,n\}$, write
$x_S:=\prod_{i\in S}x_i$.

\begin{theorem}[Diagonal normal form~{\cite[Theorem~3]{CuiGottesmanKrishna2017}}]
\label{thm:diagonal-normal-form}
Let $S$ range over the nonempty subsets of $\{1,\ldots,n\}$ with
$|S|\leq k$, and choose integers $0\leq a_S<2^{k-|S|+1}$. Up to a global
phase, the $n$-qubit diagonal elements of $C_k$ are precisely the
unitaries defined by
\begin{equation}
\label{eq:diagonal-normal-form}
    p(x)=\sum_S 2^{|S|-1}a_Sx_S,
    \qquad
    D\ket{x}=e^{2\pi i p(x)/2^k}\ket{x}.
\end{equation}
\end{theorem}

A constant term in $p$ contributes only an overall phase and is omitted.
To see that the remaining coefficients are unique, evaluate $p(x)$
modulo $2^k$ on strings whose entries are $1$ exactly at the positions
in $S$, considering these sets in increasing order of size.
Once the coefficients for proper subsets of $S$ are known, subtracting
their contributions leaves $2^{|S|-1}a_S$ modulo $2^k$.
This determines $a_S$ uniquely in the specified range.

For $k=2$, the normal form expresses $D$ as a product of powers of
$S$ and $\mathrm{CZ}$ gates. At $k=3$, the corresponding gates are
$T$, controlled-$S$, and $\mathrm{CCZ}$, where
$T=\operatorname{diag}(1,e^{i\pi/4})$.
More generally, for a monomial of degree $d=|S|$ and integers $a$
and $r\geq1$, the phase $\exp(2\pi i a x_S/2^r)$ defines a diagonal
gate in $C_k$ whenever $r+d-1\leq k$.
The hierarchy level thus bounds both the degree of each monomial
and the denominator needed to express its phase.

The gate set consists of $G_c=\{H,S,\CNOT\}$ together with powers of
\[
    R_k:=\begin{pmatrix}1&0\\ 0&e^{2\pi i/2^k}\end{pmatrix}.
\]
For this fixed $k$, write
\[
    G:=G_c\cup\{R_k^a:a=0,\ldots,2^k-1\}.
\]
All circuits in this subsection use the finite gate set $G$, with each
power of $R_k$ counted as a single gate. As in the Clifford construction,
two-qubit gates may act on any pair of qubits, and gates in the same
layer act on disjoint qubits. We add bit strings over $\mathbb F_2$
and compare integer phase numerators modulo $2^r$ when the denominator
is $2^r$.

\subsubsection{Implementing the normal form I: one monomial}
For a single monomial, the controlled-unitary construction of
Barenco et al.~\cite[Lemma~7.5 and Corollary~7.6]{BarencoEtAl1995}
gives a circuit whose size is quadratic in the monomial degree.
We give the construction explicitly to verify that every required
rotation belongs to $G$.

\begin{lemma}[Local phase synthesis]
\label{lem:local-phase-synthesis}
Let $d,r$ be positive integers satisfying $r+d-1\leq k$, and let $a$ be
an integer. For bits $u_1,\ldots,u_d$, define
\[
    Q\ket{u_1,\ldots,u_d}
    =\exp\!\left(\frac{2\pi i a}{2^r}\prod_{\ell=1}^d u_\ell\right)
      \ket{u_1,\ldots,u_d}.
\]
Then $Q$ has an exact circuit over $G$ of size and depth $O(d^2)$, with
no additional qubits.
\end{lemma}

\begin{proof}
Writing the phase angle as $\theta=2\pi a/2^r$, we implement the case
$d=1$ with a single rotation. For $d=2$, set $u=u_1$ and $v=u_2$
and use the identity
\[
    2uv=u+v-(u\oplus v).
\]
The first two terms give phase rotations of angle $\theta/2$ on
$u$ and $v$. To implement the last term, compute $u\oplus v$ into
$v$ with a $\CNOT$ gate, apply a phase rotation of angle $-\theta/2$,
and undo the $\CNOT$. This gives three rotations and two $\CNOT$ gates.

The construction for $d\geq3$ also uses Toffoli gates. In this case,
$r+d-1\leq k$ implies $k\geq3$, so these gates can be implemented
exactly over $G$. To see this, for bits $u,v,w$, the identity
\[
    4uvw=u+v+w-(u\oplus v)-(u\oplus w)-(v\oplus w)
          +(u\oplus v\oplus w)
\]
expresses the phase of $\mathrm{CCZ}$ as a sum of parity terms with
angles $\pm\pi/4$. Each parity is computed into one of the three
qubits using $\CNOT$ gates and uncomputed after the phase is applied.
This requires no additional qubits, and all the rotations are powers
of $R_k$. Applying $H$ to the third qubit before and after
$\mathrm{CCZ}$ gives an exact Toffoli gate.

To reduce the monomial degree from $d$ to $d-1$, set
\[
    f=u_1\cdots u_{d-2},\qquad u=u_{d-1},\qquad v=u_d,
\]
and use
\begin{equation}
\label{eq:recursive-monomial-phase}
    2fuv=fv+uv-(u\oplus f)v.
\end{equation}
The transformation $u\mapsto u\oplus f$ is a NOT gate controlled by
the first $d-2$ qubits. It can be implemented using $O(d)$ Toffoli
and $\CNOT$ gates, with $v$ as a dirty auxiliary
qubit~\cite[Corollary~7.4]{BarencoEtAl1995}.
The circuit leaves the controls unchanged and restores $v$ exactly.
Replacing each Toffoli gate by the construction above gives a
circuit over $G$.

Apply this operation, then the phase $\exp(-i\theta uv/2)$, and then
apply the operation again. Next apply the phases $\exp(i\theta uv/2)$
and $\exp(i\theta fv/2)$. Equation~\eqref{eq:recursive-monomial-phase}
shows that the total phase is $\exp(i\theta fuv)$. The use of $v$ as a
dirty qubit does not affect this calculation, since each controlled
operation restores $v$ before a phase is applied and the second operation
also restores $u$.

The two-qubit phases use the $d=2$ construction, and the phase involving
$fv$ is implemented recursively on $d-1$ qubits. At each step, the degree
falls by one and the angle is halved. The one-qubit rotation angles are
integer multiples of $2\pi/2^{r+d-1}$, and hence of $2\pi/2^k$. Thus all
rotations are powers of $R_k$. Each step adds $O(d)$ gates and uses only
the original $d$ qubits, giving size and depth $O(d^2)$ after summing over
the degrees.
\end{proof}

We next apply the local circuit to a work register and compare its phase
before and after adding the input. Repeating this comparison over a cycle
of binary linear maps cancels the dependence on the initial work state.

\begin{lemma}
\label{lem:catalytic-phase-gadget}
Let $d,r$ be positive integers satisfying $r+d-1\leq k$, and let $a$ be
an integer. For $x\in\{0,1\}^d$, define
\[
    D\ket{x}
    =\exp\!\left(\frac{2\pi i a}{2^r}\prod_{\ell=1}^d x_\ell\right)
      \ket{x}.
\]
For every integer $t\geq d$ satisfying $2^r\mid a2^t$, there is a circuit
$\mathcal C$ over $G$ using $t$ catalytic qubits and no clean qubits
such that
\[
    \mathcal C=D\otimes I_W.
\]
Its size and depth are $O(t^2 2^t)$. In particular,
$t\geq\max\{d,r\}$ always suffices.
\end{lemma}

\begin{proof}
Let $I$ be the $d$-qubit input register and identify the basis labels of
$W$ with $\mathbb F_2^t$. Choose an invertible $t\times t$ binary matrix
$A$ whose powers cycle through all nonzero vectors. Such a matrix
represents multiplication by a generator of the multiplicative group of
a field with $2^t$ elements, written in a binary basis. It satisfies
$A^{2^t-1}=I$ and,
for every $z\neq 0$,
\[
    \{A^jz:0\leq j\leq 2^t-2\}=\mathbb F_2^t\setminus\{0\},
\]
with each nonzero vector occurring once.

For $t=2$, we may take $A=P$, where $P$ is the matrix used in
Lemma~\ref{lem:roots-of-unity-trick}. Its powers cycle through the
three nonzero vectors of $\mathbb F_2^2$. The choice of $A$ extends
this property to work blocks of arbitrary size.

For $z\in\mathbb F_2^t$, set $f(z):=\prod_{\ell=1}^d z_\ell$, and write
$\bar x:=(x_1,\ldots,x_d,0,\ldots,0)\in\mathbb F_2^t$.
Let $F$ consist of $d$ disjoint $\CNOT$ gates with action
\[
    F\ket{x}_I\ket{z}_W=\ket{x}_I\ket{z+\bar x}_W.
\]
The zero coordinates of $\bar x$ specify targets that are unchanged by
$F$, so they impose no initialisation condition on the work qubits.
The two remaining operations act on $W$:
\[
    Q\ket{z}=\exp\!\left(\frac{2\pi i a}{2^r}f(z)\right)\ket{z},
    \qquad
    \mathcal A\ket{z}=\ket{Az}.
\]
Extending these operators by the identity on $I$, define
\[
    \mathcal C=(\mathcal A QFQ^{-1}F)^{2^t-1},
\]
where the rightmost factor acts first. One round has action
\[
    \mathcal A QFQ^{-1}F\ket{x}_I\ket{z}_W
    =\exp\!\left(\frac{2\pi i a}{2^r}
        \bigl(f(z)-f(z+\bar x)\bigr)\right)\ket{x}_I\ket{Az}_W.
\]
After $j$ rounds the work label is $A^jz$, and after all $2^t-1$ rounds
it returns to $z$.

We evaluate the sum of differences in the exponent separately for zero
and nonzero work strings. If $z=0$, each difference equals $-f(\bar x)$
because $f(0)=0$. For $z\neq 0$, the powers of $A$ visit every nonzero work
string, and translation by $\bar x$ permutes $\mathbb F_2^t$, giving
\[
    \sum_{y\in\mathbb F_2^t}\bigl(f(y)-f(y+\bar x)\bigr)=0.
\]
Removing the term $y=0$ leaves $f(\bar x)$, so the two cases give
\begin{equation}
\label{eq:phase-cycle-sum}
    \sum_{j=0}^{2^t-2}\bigl(f(A^jz)-f(A^jz+\bar x)\bigr)
    =
    \begin{cases}
        (1-2^t)f(\bar x),& z=0,\\[1ex]
        f(\bar x),& z\neq 0.
    \end{cases}
\end{equation}
After multiplication by $a$, both expressions are congruent to
$af(\bar x)$ modulo $2^r$, because $2^r\mid a2^t$. Thus the circuit has
exactly the phase of $D$ on every joint basis state. By linearity,
$\mathcal C=D\otimes I_W$.

Binary row elimination implements $\mathcal A$ with $O(t^2)$ $\CNOT$
gates, while each of $Q$ and $Q^{-1}$ uses $O(d^2)$ gates by
Lemma~\ref{lem:local-phase-synthesis}, without further workspace.
Including the two applications of $F$, one round has size and depth
$O(t^2)$, and the $2^t-1$ rounds give the claimed bounds.
\end{proof}

The operator identity also preserves entanglement with an external
register. In choosing $t$, one may first cancel common powers of two
from $a/2^r$ and omit coefficients that vanish modulo $2^r$.

\subsubsection{Implementing the normal form II: full expression}
The full diagonal unitary is the product of its monomial phase gates.
We implement these gates in parallel using separate work blocks.
For one monomial, the input addition $F$ in
Lemma~\ref{lem:catalytic-phase-gadget} consists of disjoint $\CNOT$
gates. Across several monomials, the same input qubit may control
targets in different work blocks. The following lemma implements
all gates sharing one control in logarithmic depth, without requiring
the targets to be initialised. We use it below to implement the
combined input addition, also denoted by $F$.

\begin{lemma}[Parallel $\CNOT$ gates with a common control]
\label{lem:dirty-target-fanout}
Let $b\geq 1$ be an integer and let $x,z_1,\ldots,z_b$ be bits. The
transformation
\[
    \ket{x}\ket{z_1,\ldots,z_b}
    \longmapsto\ket{x}\ket{z_1+x,\ldots,z_b+x}
\]
has a $\CNOT$ implementation of depth $2\lceil\log_2 b\rceil+1$ and
size $2b-1$, with no additional qubits. The targets may be in an arbitrary
state.
\end{lemma}

\begin{proof}
Let $e_1=(1,0,\ldots,0)^{\mathsf T}$ and
$\mathbf 1=(1,\ldots,1)^{\mathsf T}$. On the $b$ target qubits, construct
a binary tree of $\CNOT$ gates whose matrix $L$ satisfies
$Le_1=\mathbf 1$. Starting from the first target, each layer uses the
qubits already reached as controls for distinct new targets. This requires
$b-1$ gates in $\lceil\log_2 b\rceil$ layers and defines an invertible
linear map on all target strings.

Write $z=(z_1,\ldots,z_b)^{\mathsf T}$. Apply the inverse target circuit,
then one $\CNOT$ from $x$ to the first target, and then the target circuit
itself. The resulting string is
\[
    L(L^{-1}z+xe_1)=z+x\mathbf 1.
\]
The gate count and depth follow. For $b=1$, the target circuit is empty.
\end{proof}

\begin{lemma}[Parallel implementation of monomial phases]
\label{lem:catalytic-diagonal-monomials}
Let $D$ be given by~\eqref{eq:diagonal-normal-form}, with its nonzero
coefficients specified. Let $s\geq 1$ be an integer, and let $N$ be the
number of nonzero monomials.

There is a circuit $\mathcal C_D$ over $G$ using at most
$k\min\{N,sn\}$ catalytic qubits and no clean qubits such that
\[
    \mathcal C_D=D\otimes I_W.
\]
Its depth and size are, respectively,
\[
    O\!\left(2^k\left\lceil\frac{N}{sn}\right\rceil
    \bigl(k^2+\log(\min\{N,sn\}+1)\bigr)\right)
    \qquad\text{and}\qquad O(k^2 2^kN).
\]
In particular, for fixed $k$ its depth is
\[
    O\!\left(\left(1+\frac{N}{sn}\right)\log(n+1)\right),
\]
its size is $O(N)$, and it uses $O(sn)$ catalytic qubits.
If $sn\geq N$, its depth is $O(\log(n+1))$, and its size and
workspace are $O(n^k)$.
\end{lemma}

\begin{proof}
The case $N=0$ is the identity and requires no gates. Otherwise,
partition the nonzero monomials into $\lceil N/(sn)\rceil$ batches,
each containing at most $sn$ monomials. Apply the construction to
each batch in turn, reusing the same work register.

Consider a batch containing $b$ monomials, and assign a separate
block of $k$ work qubits to each monomial. For a support
$S=\{i_1,\ldots,i_d\}$, the coefficient in $p(x)/2^k$ is
$a_S/2^{k-d+1}$. Use Lemma~\ref{lem:catalytic-phase-gadget} with
$t=k$, $r=k-d+1$, and $a=a_S$. Order each support once and use this
order throughout the circuit. The input $\CNOT$ gates add
$x_{i_1},\ldots,x_{i_d}$ to the first $d$ targets in its work block.
All blocks use the same matrix $A$ and the same $2^k-1$ rounds.

Let $F$ denote the product of these input $\CNOT$ gates over all
blocks in the batch. Each input bit controls at most $b$ targets.
Implement these gates using
Lemma~\ref{lem:dirty-target-fanout}. Since different input bits have
disjoint target sets and distinct controls, their circuits run in
parallel. A complete application of $F$ has depth
$O(\log(b+1))$ and size $O(kb)$.

Each round applies $F$, all $Q^{-1}$ circuits, $F$ again, all $Q$
circuits, and finally all $\mathcal A$ circuits, in this order.
The phase operations and binary linear transformations act on
disjoint work blocks and begin only after the preceding application
of $F$ is complete. Although the fanout circuit uses gates between
different blocks, it implements $F$ exactly.

After all rounds, every work label is restored. By
\eqref{eq:phase-cycle-sum}, each block contributes its monomial
phase independently of the initial work string. The circuit for
the batch thus applies the required phases and acts as the identity
on the work register. This register can be reused for the next
batch. The product of the phases over all batches is
$\exp(2\pi i p(x)/2^k)$, giving $\mathcal C_D=D\otimes I_W$.

In each round, the phase operations and binary linear transformations
have depth $O(k^2)$ and size $O(k^2b)$. Including the two applications
of $F$, the round has depth
$O(k^2+\log(b+1))$ and size $O(k^2b)$.
Multiplying by $2^k-1$ rounds and $\lceil N/(sn)\rceil$ batches
gives the depth bound. The batch sizes sum to $N$, so the total
size is $O(k^2 2^kN)$. Reusing the work register requires only
$k$ qubits for each monomial in the largest batch, giving at most
$k\min\{N,sn\}$ catalytic qubits.

Finally,
\[
    N\leq\sum_{d=1}^{\min(k,n)}\binom nd=O(n^k)
\]
for fixed $k$. Together with $\lceil N/(sn)\rceil\leq 1+N/(sn)$,
this gives the remaining bounds.
\end{proof}

The work register can be smaller for a given list of coefficients.
Reduce each nonzero fraction $a_S/2^{k-|S|+1}$ to an odd numerator
and denominator $2^{r_S}$, and choose $t_S=\max\{|S|,r_S\}$.
Within each batch, round $j=0,1,\ldots$ acts only on blocks with
$j<2^{t_S}-1$, excluding completed blocks from the fanout target
sets. The workspace for a batch is the sum of its block sizes.
Since the register is reused, the largest of these sums suffices
for the full circuit. Each $t_S$ is at most $k$, so the same
worst-case bounds hold.

Terms with $|S|=1$ can instead be applied directly to the input
in one layer, without any work qubits. The matrix $A$ for each
block size can be chosen once for fixed $k$, and the circuit is
then explicit from the coefficients and supports.

If the operation in Lemma~\ref{lem:dirty-target-fanout} is available
as a single native gate, including on targets in an arbitrary state,
each application of $F$ has constant depth. For fixed $k$, the
construction above then has depth $O(\lceil N/(sn)\rceil)$.
In particular, implementing all $N$ monomials in parallel gives
constant depth using at most $kN$ catalytic qubits. The construction
below uses less workspace but also requires catalytic circuits for
binary linear maps, whose depth must be accounted for separately.

\subsubsection{Reducing the workspace}
\label{sec:diagonal-batching}

The workspace count is dominated by the monomials with the highest degree, i.e. those that apply a phase $-1$ depending on $k$ input qubits. We can construct an optimal compiler for these phases using the $\CNOT$ compiler developed above. We first specify the specific diagonal terms:
\begin{definition}\label{def:many_constrol_phase}
Let $k\geq 2$, let $S$ range over subsets of $\{1,\dots, n\}$ with $|S| = k$, and choose bits $a_S\in\{0,1\}$. Define the diagonal unitary $D_k$ by
\[
    p(x) = \sum_{S} a_S x_S, \qquad D_k\ket{x} = (-1)^{p(x)}\ket{x},
\]
where $p(x)$ is evaluated over $\mathbb F_2$.

\end{definition}

To optimize these terms we take inspiration from~\cite{heyfron2018efficientquantumcompilerreduces}, who group degree-three terms by a shared control to obtain controlled quadratic phases. We extend this grouping to degree $k$ and reduce each quadratic with Dickson's theorem~\cite{Dic01}.

The main idea is to batch different terms together, splitting the full polynomial into a sum of degree-$(k-2)$ monomials, each multiplied by a quadratic form. We can then use Dickson's theorem to optimize each quadratic form separately. This reduces the number of terms in each quadratic form from $O(n^2)$ to $O(n)$ monomials, so the total number of monomials after the basis changes is $O(n^{k-1})$, at the cost of applying these basis changes. Each basis change consists of a $\CNOT$ circuit followed by local $X$ gates. The number of basis changes equals the number of degree-$(k-2)$ monomials, which is at most $\binom{n}{k-2} = O(n^{k-2})$.

All basis changes can be run in parallel, and one basis change costs $O(n^2/\log^2 n)$ catalytic qubits at depth $O(\log n)$. The total cost is therefore $O(n^k/\log^2 n)$ catalytic qubits, which is the dominant cost in the new circuit. Compared with applying Lemma~\ref{lem:catalytic-diagonal-monomials} directly to all $O(n^k)$ monomials, the efficient catalytic $\CNOT$ circuits therefore reduce the workspace for the high-degree phases by a factor $\log^2 n$.

We first introduce, and prove, Dickson's theorem:
\begin{lemma}[Dickson's Theorem \cite{Dic01}]\label{lem:dicksons}
Let $q: \mathbb F^{n}_2 \rightarrow \mathbb F_2$ be a polynomial of degree at most $2$ with $q(0)=0$, and write it as:
\[
    q(x) = \sum_{0\leq i< j\leq n-1}B_{ij}x_ix_j + \sum_{i=0}^{n-1}A_i x_i.
\]
Then there exists an affine transformation, given by a matrix $M \in GL(n,2)$ and a vector $c \in \mathbb F^{n}_2$ such that $y = Mx + c$, and an integer $d$ with $2d \leq n$, such that
\[
    q(x) = \sum_{i=0}^{d-1}y_{2i}y_{2i+1} + \delta y_{2d} + \epsilon,
\]
where $\delta, \epsilon \in \{0,1\}$, at most one of them is $1$, and the linear term is omitted if $2d = n$.
\end{lemma}

\begin{proof}

The goal is to construct a basis transformation that constructs the right $y$ values. If $q(x)$ contains a quadratic term, pick one such term $x_ux_v$ and relabel the variables such that it becomes $x_0x_1$, so $B_{01} = 1$. A relabelling is a permutation of the coordinates, which is also an element of $GL(n,2)$. Now look at all terms involving $x_0$ and $x_1$. These can be written as:
\[
    q_{x_0,x_1}(x) = x_0x_1 + x_0 l_{\geq2}(x) + x_1 k_{\geq 2}(x) + ax_0 + bx_1,
\]
where $l_{\geq 2}(x)$ is the sum of the $x_i$ such that $B_{0i} = 1$ for $i\geq 2$, $k_{\geq 2}(x)$ the same such that $B_{1i} = 1$ for $i \geq 2$, and $a = A_0$, $b = A_1$. Now let
\[
    y_0 = x_0 + k_{\geq 2}(x) + b, \qquad y_1 = x_1 + l_{\geq 2}(x) + a,
\]
then it follows that,
\[
    y_0 y_1 = x_0 x_1 + x_0 l_{\geq2}(x) + x_1 k_{\geq2}(x) + ax_0 + bx_1 + l_{\geq2}(x)k_{\geq2}(x) + a k_{\geq2}(x) + b l_{\geq2}(x) + ab.
\]
Importantly the additional terms $l_{\geq2}(x)k_{\geq2}(x) + a k_{\geq2}(x) + b l_{\geq2}(x)$ contain at most quadratic terms in $x_i$ such that $i \geq 2$. Additionally notice that it might happen that in the product $l_{\geq2}(x)k_{\geq2}(x)$ some term $x_i^2$ appears, but because this is over $\mathbb F_2$ it holds that $x_i^2 = x_i$, which means that it adds at most a linear term. The term $ab$ is an additive constant, which in the end will possibly produce the term $\epsilon$. The new variables $y_0$ and $y_1$ only add terms in $x_{i}$ with $i\geq 2$ and constants to $x_0$ and $x_1$, so this transformation is invertible. After this step we can now rewrite $q(x)$:
\[
    q(x) = y_0y_1 + \sum_{2\leq i < j \leq n-1} B'_{ij} x_i x_j + \sum_{i = 2}^{n-1} A'_i x_i + \epsilon',
\]
where $B'$ contains the extra quadratic terms from $l_{\geq2}(x)k_{\geq2}(x)$, $A'$ the extra linear terms produced by the additional terms, and $\epsilon' = ab$.

This step can now be repeated on the remaining variables using the updated $B'$ and $A'$. Sequentially repeat this, each time picking a remaining quadratic term, until after $d$ steps no quadratic term is left. Relabelling the paired variables as $y_0,\dots,y_{2d-1}$ gives
\[
    q(x) = \sum_{i=0}^{d-1}y_{2i}y_{2i+1} + \sum_{i = 2d}^{n-1} A''_i x_i + \epsilon',
\]
where $\epsilon'$ is the sum of all terms $ab$ modulo $2$. Notice that the linear terms on the paired variables are already absorbed by the constants $a$ and $b$ in the definition of the $y$'s. The remaining linear terms $\sum_{i = 2d}^{n-1} A''_i x_i$ only involve variables that do not appear in any product, so they cannot be absorbed in the same way.

If all $A''_i = 0$, we are done with $\delta = 0$ and $\epsilon = \epsilon'$. Otherwise pick an index $j \geq 2d$ with $A''_j = 1$ and relabel such that $j = 2d$. Now let
\[
    y_{2d} = \sum_{i = 2d}^{n-1} A''_i x_i + \epsilon',
\]
and keep $y_i = x_i$ for the other unpaired variables. This is invertible because $y_{2d}$ contains $x_{2d}$. All remaining linear terms and the constant are now collected in the single variable $y_{2d}$, which gives $\delta = 1$ and $\epsilon = 0$. This last linear term cannot be removed. If it is present, $q$ equals $1$ on exactly $2^{n-1}$ inputs, while a sum of products plus a constant never does, and an affine transformation does not change this number.

All steps are invertible affine transformations, so their composition is of the form $y = Mx + c$ with $M \in GL(n,2)$, which proves the lemma.
\end{proof}

We will now show how this can be used to implement the phase circuits $D_k$.

\begin{lemma}
\label{lem:batched-high-degree-phases}
Fix $k\geq 2$, and let $s\geq 1$ be an integer.
Every $D_k$ in Definition~\ref{def:many_constrol_phase} has an exact
catalytic implementation over $G$ using $O(sn)$ catalytic qubits
and no clean qubits, with depth
\[
    O\!\left(\log(n+1)+\frac{n^{k-1}}{s\log(n+1)}\right)
\]
and size
\[
    O\!\left(\frac{n^k}{\log(n+1)}\right).
\]
In particular, depth $O(\log(n+1))$ requires at most
$O(n^k/\log^2(n+1))$ catalytic qubits.
\end{lemma}

\begin{proof}

We start by decomposing the polynomial $p(x)$, where we write $x_T = \prod_{i \in T} x_i$. Assign each nonzero monomial to one of its subsets of size $k-2$, so that every monomial occurs in exactly one group. This gives
\[  
    p(x) = \sum_{T\in \mathcal T}x_Tq_T(x),\qquad g:= |\mathcal T|\leq \binom{n}{k-2} = O(n^{k-2}),
\]
where the $q_T(x)$ are homogeneous quadratic polynomials involving only variables outside $T$. When $k=2$, take $T=\emptyset$ and $x_T=1$. If $g=0$, the circuit is the identity. Any $q_T(x)$ consists naively of $O(n^2)$ monomials, so a naive implementation of $p(x)$ requires evaluating $O(n^k)$ different monomials. However, because $q_T(x)$ is a quadratic form, Lemma~\ref{lem:dicksons} gives the construction of a basis transformation that reduces the number of monomials in $q_T(x)$ to at most $\lfloor n/2 \rfloor + 1$. The additional term is either linear or constant. Let $y_T = M_Tx + c_T$, for $M_T \in GL(n,2)$ and $c_T \in \mathbb F^n_2$ and,
\[  
    \tilde{q}_T(y_T) = \sum_{i=0}^{r_T-1} y_{2i,T}y_{2i+1,T} + \delta_T y_{2r_T,T} + \epsilon_T,    
\]
with $r_T \leq n/2$, such that $q_T(x) = \tilde{q}_T(y_T)$ as in Lemma~\ref{lem:dicksons}. The linear term is omitted when $2r_T=n$.

Then if one has access to $y_T$ instead of $x$, the term $x_T\tilde{q}_T(y_T)$ consists of at most $r_T + 1$ monomials instead of $O(n^2)$.

We cannot write $y_T$ into a clean register, since all work qubits are catalytic. Writing it into a catalytic register gives $z + y_T$ for an unknown $z$, and applying the phase to that register would give $x_T\tilde{q}_T(z + y_T)$, which depends on $z$. Instead we add $y_T$ directly into the work blocks of the phase construction of Lemma~\ref{lem:catalytic-diagonal-monomials}. That construction gives every monomial its own block of $k$ catalytic qubits. It repeatedly applies a map $F$ that adds the variables of the monomial to its block, $\ket{\tau} \mapsto \ket{\tau + v}$. The proof of Lemma~\ref{lem:catalytic-phase-gadget} only uses that the label $v$ is the same in every round, not that it consists of input bits. For $T=\{i_1,\ldots,i_{k-2}\}$ and the monomial $x_T y_{2i,T} y_{2i+1,T}$, use the label
\[
    v=(x_{i_1},\ldots,x_{i_{k-2}},y_{2i,T},y_{2i+1,T}),
\]
whose last two entries are affine functions of $x$. Each factor of $x_T$ occupies a separate coordinate. The linear term $\delta_T x_T y_{2r_T,T}$ and the constant term $\epsilon_T x_T$ use the corresponding shorter lists, padded with zeros to length $k$.

For a term with $d\geq 1$ factors, use Lemma~\ref{lem:catalytic-phase-gadget} with $t=k$, $r=1$, and $a=1$. The local phase $Q$ acts on the first $d$ coordinates. For each fixed $x$, the same vector $v$ is added in every round, so~\eqref{eq:phase-cycle-sum} gives the required product independently of the initial work string. No independence assumption on the entries of $v$ is needed. If $k=2$ and $\epsilon_T=1$, the resulting constant phase $-1$ is applied directly using $(XZ)^2=-I$ on one input qubit.

For this addition, write $O$ for the phase blocks of one group and $W$ for the additional catalytic qubits. Let $C_T$ be the circuit that adds all labels of group $T$ to their blocks:
\[
    C_T\ket{x}_I\ket{\tau}_O\ket{w}_W = \ket{x}_I\ket{\tau + L_T x + c'_T}_O\ket{w}_W,
\]
where every row of $L_T$ is either a unit vector for an index in $T$ or a row of $M_T$, with zero rows for the padding, and $c'_T$ collects the corresponding entries of $c_T$. The matrix $L_T$ has $n$ columns and $O(n)$ rows. Its linear addition is represented by the invertible matrix
\[
    \begin{pmatrix}I&0\\ L_T&I\end{pmatrix}
\]
on $O(n)$ qubits. For sufficiently large $n$ and an integer $s'$ with $1\leq s'\leq n/\log^2 n$, Theorem~\ref{thm:catalytic-cnot} implements this addition in depth $O(n/(s'\log n))$ using $O(s'n)$ additional catalytic qubits.

All $C_T$ read the same input. To run them in parallel, give each active group an additional $n$-qubit register with arbitrary basis label $u_T$. First add $x$ to every $u_T$, apply $L_T$ from each copy to its phase blocks, add $x$ to every $u_T$ again, and apply $L_T$ once more. The copies return to $u_T$, and the phase labels become
\[
    \tau+L_T(u_T+x)+L_Tu_T=\tau+L_Tx.
\]
The unknown content of each copy enters twice and cancels because the map is linear. Now add $c'_T$ by local $X$ gates. The affine constant is added once, after the two linear additions. Each linear addition restores its additional catalytic qubits, so this implements $C_T$ exactly and restores the copies and all other work qubits outside the phase blocks.

For $b$ active groups, each addition of $x$ to the copies has depth $O(\log(b+1))$ and size $O(nb)$ by Lemma~\ref{lem:dirty-target-fanout}. The linear circuits act on disjoint registers and run in parallel.

We now allocate the available workspace between these groups. For sufficiently large $n$, choose
\[
    s':=\min\{s,\lfloor n/\log^2 n\rfloor\},
    \qquad
    b:=\min\{g,\lfloor s/s'\rfloor\}.
\]
Process at most $b$ groups at a time, reusing the work registers between batches. Since $bs'\leq s$, the phase blocks, the copies, and the additional registers for the linear circuits together use $O(sn)$ catalytic qubits.

We run the construction of Lemma~\ref{lem:catalytic-diagonal-monomials} on the polynomial
\[
    p(x) = \sum_{T \in \mathcal T} x_T \tilde{q}_T(y_T),
\]
which has $O(n^{k-1})$ products in the indicated coordinates, with $F$ implemented by the maps $C_T$ for the active groups. The phase blocks use $O(kn)$ catalytic qubits per group. Each round applies $F$, the local $Q^{-1}$ circuits, $F$ again, the local $Q$ circuits, and the binary linear maps $\mathcal A$, in this order. Each application of $F$ is completed before the next local operation begins. There are $2^k-1$ rounds, after which every work register is restored and can be reused for the next batch.

Since $k$ is fixed, $b\leq g=O(n^{k-2})$ and $n/(s'\log n)\geq\log n$, a batch has depth $O(n/(s'\log n))$, including the additions to the copies. The total depth is
\[
    O\!\left(\left\lceil\frac{g}{b}\right\rceil
        \frac{n}{s'\log n}\right)
    =O\!\left(\log n+\frac{gn}{s\log n}\right).
\]
Indeed, when $s\leq\lfloor n/\log^2 n\rfloor$, we have $s'=s$ and $b=1$: the groups are processed one at a time, using the available workspace for each linear circuit. For larger $s$, each group has depth $O(\log n)$, and additional workspace allows more groups to run in parallel. Using $g=O(n^{k-2})$ gives the depth bound.

Each linear addition has total width $O(s'n)$ and depth $O(n/(s'\log n))$, hence size $O(n^2/\log n)$. The additions to the copies and the operations on the phase blocks contribute $O(n)$ gates per group per round. Summing over the groups and the $2^k-1$ rounds gives size $O(gn^2/\log n)=O(n^k/\log n)$.

In particular, choosing $s$ of order $n^{k-1}/\log^2 n$ allows all groups to run in parallel. The workspace is dominated by the maps $C_T$, giving
\[
    \text{depth } O(\log n), \qquad
    \text{workspace } O\!\left(\frac{n^k}{\log^2 n}\right).
\]
The remaining finitely many values of $n$ can be implemented directly. Writing $\log(n+1)$ in the bounds makes them valid for all $n\geq 1$.
\end{proof}

\subsubsection{Final circuit}
The degree-$k$ terms in the diagonal normal form give the phases in Definition~\ref{def:many_constrol_phase}. There are fewer terms of smaller degree, and Lemma~\ref{lem:catalytic-diagonal-monomials} suffices to implement them within the same bounds.

\begin{theorem}[Catalytic implementation of diagonal elements]
\label{thm:catalytic-diagonal}
Fix $k\geq 2$, and let $s\geq 1$ be an integer. Let $D$ be given by~\eqref{eq:diagonal-normal-form}, with its nonzero coefficients specified. There is a circuit $\mathcal C_D$ over $G$ using $O(sn)$ catalytic qubits and no clean qubits such that
\[
    \mathcal C_D=D\otimes I_W.
\]
Its depth and size are, respectively,
\[
    O\!\left(\log(n+1)+\frac{n^{k-1}}{s\log(n+1)}\right)
    \qquad\text{and}\qquad
    O\!\left(\frac{n^k}{\log(n+1)}\right).
\]
In particular, there is an implementation of depth $O(\log(n+1))$ using $O(n^k/\log^2(n+1))$ catalytic qubits, with the same size bound.
\end{theorem}

\begin{proof}
Apply Lemma~\ref{lem:batched-high-degree-phases} to the degree-$k$ terms. There are only $O(n^{k-1})$ remaining monomials, so Lemma~\ref{lem:catalytic-diagonal-monomials}, with the same parameter $s$, implements them in depth
\[
    O\!\left(\log(n+1)+\frac{n^{k-2}\log(n+1)}{s}\right)
\]
and size $O(n^{k-1})$, using $O(sn)$ catalytic qubits. Since $\log^2(n+1)=O(n)$, these costs are within the bounds for the degree-$k$ terms. Both circuits act as the identity on the work register, so they can reuse it. Their product is $D\otimes I_W$, and their depths and sizes add. Choosing $s$ of order $n^{k-1}/\log^2(n+1)$ gives the logarithmic-depth implementation.
\end{proof}

\subsubsection{Lower bounds}
The following bounds show that logarithmic depth is necessary and that the workspace in Theorem~\ref{thm:catalytic-diagonal} at this depth is optimal in the worst case, for fixed $k$.

\begin{lemma}
\label{lem:diagonal-lower-bounds}
Fix $k\geq 2$ and let $s$ be an integer with $1\leq s\leq n^{k-1}/\log^2(n+1)$. Suppose every $n$-qubit diagonal unitary in~\eqref{eq:diagonal-normal-form} has an exact catalytic implementation over the finite gate set $G$, of depth at most $h$ and using at most $sn$ catalytic qubits. Then
\[
    h\geq\lceil\log_2 n\rceil,
    \qquad
    h=\Omega\!\left(\frac{n^{k-1}}{s\log(n+1)}\right).
\]
In particular, the depth in Theorem~\ref{thm:catalytic-diagonal} is optimal up to constant factors, and depth $O(\log(n+1))$ requires $\Omega(n^k/\log^2(n+1))$ catalytic qubits in the worst case.
\end{lemma}

\begin{proof}
The logarithmic depth is necessary in the worst case, even with an
arbitrarily large catalytic register. Indeed,
\[
    D=\prod_{j=2}^n\mathrm{CZ}_{1,j}\in C_2\subseteq C_k
    \qquad\text{satisfies}\qquad
    DX_1D^\dagger=X_1\prod_{j=2}^n Z_j.
\]
A layer of disjoint two-qubit gates at most doubles the support of an
operator. A circuit implementing $D\otimes I_W$ must therefore have
depth at least $\lceil\log_2 n\rceil$.

A counting argument gives a corresponding comparison for workspace.
For fixed $k$, uniqueness of the coefficients in
\eqref{eq:diagonal-normal-form} gives
\[
    2^{\sum_{d=1}^{\min(k,n)}(k-d+1)\binom nd}
    =2^{\Theta(n^k)}
\]
distinct diagonal unitaries. The constants in this argument may depend
on $k$. A layer on $w$ qubits over the finite gate set $G$ has at most
$\exp(O(w\log(w+1)))$ possible choices. Thus, to implement every
such diagonal element exactly with depth at most $h$ and width at most
$w$, we need
\[
    hw\log(w+1)=\Omega(n^k).
\]
With $sn$ catalytic qubits, the total width is $w=(1+s)n\leq 2sn$, which
gives the stated depth bound. For $s\leq\mathrm{poly}(n)$ we have
$\log(n+sn+1)=O(\log(n+1))$, so this matches the depth in
Theorem~\ref{thm:catalytic-diagonal}. At depth $h=O(\log(n+1))$, it gives the
worst-case lower bound $sn=\Omega(n^k/\log^2(n+1))$. If $w>n^k$ this is
immediate, and otherwise $\log(w+1)=O(\log n)$.

\end{proof}

\subsection{Semi-Clifford circuits}
The diagonal construction also applies to unitaries that become diagonal
after multiplication by Clifford operators, known as semi-Clifford
unitaries.

\begin{definition}[Semi-Clifford unitary]
An $n$-qubit unitary is semi-Clifford if conjugation by it maps some maximal
abelian subgroup of $P_n$ to another maximal abelian subgroup of $P_n$.
\end{definition}

Equivalently, a unitary $U$ is semi-Clifford when there are Clifford unitaries $\mathcal C_1,\mathcal C_2$ and a diagonal unitary $D$
such that
\[
    U=\mathcal C_1D\mathcal C_2,
\]
as shown in~\cite[Proposition~1]{ZengChenChuang2008}. We call the Clifford unitaries $\mathcal C_1$ and $\mathcal C_2$ in
$U=\mathcal C_1D\mathcal C_2$ the Clifford factors of the decomposition.
If $U\in C_k$, then $D=\mathcal C_1^\dagger U\mathcal C_2^\dagger$ also belongs to $C_k$, since multiplication on either side by a Clifford unitary preserves membership in $C_k$~\cite[Proposition~3]{ZengChenChuang2008}. Any overall phase of $D$ can be absorbed into $\mathcal C_1$, so we
may write $D$ in the form~\eqref{eq:diagonal-normal-form}.

\begin{corollary}[Catalytic implementation of semi-Clifford gates]
\label{cor:catalytic-semi-clifford}
Fix $k\geq 2$, and let $s\geq 1$ be an integer.
Let an $n$-qubit semi-Clifford unitary $U\in C_k$ be given by a
decomposition $U=\mathcal C_1D\mathcal C_2$, with the two Clifford
factors and the coefficients of $D$ specified. There is a
catalytic implementation over $G$ using $O(sn)$ catalytic qubits
and no clean qubits, with depth
\[
    O\!\left(\log(n+1)+\frac{n^{k-1}}{s\log(n+1)}\right)
\]
and size
\[
    O\!\left(\frac{n^k}{\log(n+1)}\right).
\]
It acts as $U\otimes I_W$ up to one global phase independent of
the input and catalyst.

In particular, there is an implementation of depth
$O(\log(n+1))$ using $O(n^k/\log^2(n+1))$ catalytic qubits,
with the same size bound.
\end{corollary}

\begin{proof}
For sufficiently large $n$, Lemma~\ref{cor:catalytic-clifford}, with parameter
\[
    s':=\min\{s,\lfloor n/\log^2 n\rfloor\},
\]
gives a circuit for each Clifford factor of depth
\[
    O\!\left(\log n+\frac{n}{s\log n}\right)
\]
using $O(s'n)$ catalytic qubits. Since gates in each layer act on disjoint qubits, multiplying the depth $O(n/(s'\log n))$ by the total width $O(s'n)$ gives size $O(n^2/\log n)$. For the remaining finitely many values of $n$, the Clifford factors can be implemented directly with constant size and depth. Writing $\log(n+1)$ in place of $\log n$ makes the bounds valid for all $n\geq 1$.

Together with the circuit for $D$ from Theorem~\ref{thm:catalytic-diagonal}, these circuits implement $U$ by applying $\mathcal C_2$, $D$, and $\mathcal C_1$ in this order. Each circuit implements the required unitary and acts as the identity on the work register, up to a global phase independent of the input and work state. The three circuits can thus share a work register large enough for the largest workspace requirement, leaving any unused qubits untouched. Their depths and sizes add. For fixed $k\geq 2$, the bounds for $D$ cover the costs of both Clifford factors, giving the claimed bounds. Choosing $s$ of order $n^{k-1}/\log^2(n+1)$ gives the logarithmic-depth implementation.
\end{proof}

For a diagonal factor with only $N$ nonzero coefficients,
Lemma~\ref{lem:catalytic-diagonal-monomials} gives another implementation.
With $O(sn)$ catalytic qubits, its depth and size are
\[
    O\!\left(
        \frac{n}{s\log(n+1)}
        +\left(1+\frac{N}{sn}\right)\log(n+1)
    \right)
    \qquad\text{and}\qquad
    O\!\left(k^2 2^kN+\frac{n^2}{\log(n+1)}\right).
\]
These bounds can be smaller when the normal form has few
nonzero coefficients.

If all $N$ monomials are implemented in parallel and each Clifford factor has depth $O(\log(n+1))$, the common work register requires only
\[
    O\!\left(\max\left\{kN,\frac{n^2}{\log^2(n+1)}\right\}\right)
\]
catalytic qubits. The maximum suffices because the three circuits use the register in turn.

\section*{Acknowledgements}
Strelchuk acknowledges support from the Wellcome Leap as part of the Q4Bio Program and the Royal Society University Research Fellowship. Subramanian acknowledges support from the Royal Society through a University Research Fellowship. 
Folkertsma acknowledges funding from the Quantum Software Consortium.

\printbibliography
\end{document}